\documentclass{svjour3}  

\usepackage[utf8]{inputenc}

\usepackage{amssymb}
\usepackage{amsmath}
\usepackage{hyperref}
\usepackage{graphicx}
\usepackage{xcolor}
\usepackage{bm}

\usepackage{cite}

\usepackage{amsmath,amssymb}
\usepackage{graphicx}
\usepackage{enumitem}
\usepackage{xspace}

\newcommand{\findmin}{\textit{FindMin}\xspace} %
\newcommand{\ftmin}{\textsc{FT-Min}\xspace} %

\newcommand{\E}{\mathbb{E}}

\smartqed 

\title{Approximate Minimum Selection with Unreliable Comparisons in Optimal Expected Time\thanks{Research supported by SNF (project number 200021\_165524).}}
\titlerunning{Apx. Minimum Selection with Unreliable Comparisons in Optimal Expected Time}

\author{Stefano Leucci \and Chih-Hung Liu}
\authorrunning{Stefano Leucci and Chih-Hung Liu}

\institute{Stefano Leucci 
	\at Department of Information Engineering, Computer Science and Mathematics, University of L'Aquila, Italy \\\email{stefano.leucci@univaq.it}  
	\and
	Chih-Hung Liu \at Department of Computer Science, ETH Z\"{u}rich, Switzerland \\\email{chih-hung.liu@inf.ethz.ch}
}

\begin{document}

\maketitle

\begin{abstract}
We consider the \emph{approximate minimum selection} problem in presence of \emph{independent random comparison faults}.
This problem asks to select one of the smallest $k$ elements in a linearly-ordered collection of $n$ elements by only performing \emph{unreliable} 
pairwise comparisons: whenever two elements are compared, there is a constant probability that the wrong answer is returned. 

We design a randomized algorithm that solves this problem  with probability $1-q \in [ \frac{1}{2}, 1)$ and for the whole range of values of $k$ using $O( \frac{n}{k} \log \frac{1}{q} )$ expected time.
Then, we prove that the expected running time of any algorithm that succeeds w.h.p. must be $\Omega(\frac{n}{k}\log \frac{1}{q})$, thus implying that our algorithm is asymptotically optimal, in expectation.
These results are quite surprising in the sense that for $k$ between $\Omega(\log \frac{1}{q})$ and $c \cdot n$, for any constant $c<1$, the expected running time must still be $\Omega(\frac{n}{k}\log \frac{1}{q})$ even in absence of comparison faults. Informally speaking, we show how to deal with comparison errors without any substantial complexity penalty w.r.t.\ the fault-free case.
Moreover, we prove that as soon as $k = O( \frac{n}{\log\log \frac{1}{q}})$, it is possible to achieve the optimal \emph{worst-case} running time of $\Theta(\frac{n}{k}\log \frac{1}{q})$. 
\end{abstract}

\section{Introduction}\label{sec-intro}

In an ideal world computational tasks are always carried out reliably, i.e., every operation performed by an algorithm behaves exactly as intended.
Practical architectures, however, are error-prone and even basic operations can sometimes return the wrong results, especially when large-scale systems are involved.
When dealing with these spurious results the first instinct is that of trying to detect and correct the errors as they manifest, so that the problems of interest can then be solved using classical (non fault-tolerant) algorithms.
An alternative approach deliberately allows errors to interfere with the execution of an algorithm, in the hope that the computed solution will still be good, at least in an \emph{approximate} sense. This begs the question: \emph{is it possible to devise algorithms that cope with faults by design and return probably good solutions?}

We investigate this question by considering a generalization of the fundamental problem of finding the minimum element in a totally-ordered set:
in the \emph{fault-tolerant approximate minimum selection} problem ($\ftmin(k)$ for short) we wish to return one of the smallest $k$ elements in a collection of size $n$ using only \emph{unreliable pairwise comparisons}, i.e., comparisons in which the result can sometimes be incorrect due to \emph{errors}.
This allows, for example, to find a representative in the top percentile of the input set, or to obtain a good estimate of the minimum from a set of noisy observations.

In this paper we provide both upper and lower bounds on the running time of any (possibly randomized) algorithm that solves $\ftmin(k)$ with probability $1-q$, where $q \in (0, \frac{1}{2}]$.
Our solution is optimal with respect to the \emph{expected} number of comparisons for the whole range of value of $k$. Moreover, we also design an algorithm that uses the asymptotically optimal number of comparisons \emph{in the worst-case}
for all values of $k$ in $O(\frac{n}{\log \log \frac{1}{q}})$.

Our results find application in any setting that is subject to random comparison errors (e.g., due to communication interferences, alpha particles, charge collection, cosmic rays \cite{Baumann05,Catania04}, or energy-efficient architectures where
the energy consumed by the computation can be substantially reduced if a small fraction of faulty results is allowed\cite{Anthes13,PalemL13,CheemavalaguKPAC05,CheemavalaguKP04}),
or in which performing accurate comparisons is too resource-consuming (think, e.g., of the elements as references to remotely stored records) while approximate comparisons can be carried out much quicker. 
One concrete application might be selecting one of the top-tier products from a collection of items obtained from an imprecise manufacturing process (i.e. a high-quality cask from a distillery or a fast semiconductor from a fabrication facility). In these settings products can be compared to one another (either by human experts or by automated tests), yet the result of the comparisons are not necessarily accurate.\footnote{Interestingly, some semiconductor companies consider the more general problem of approximately classifying all their products in a process known as ``product-binning''. For example, identically-built processors can be assigned different commercial names depending on their performances.}

Before presenting our results in more detail, let us briefly discuss the error model we use.

\subsection{The Comparison Model}
\label{sec:comparison_model}

We consider \emph{independent random comparison faults}, a simple and natural error model in which
there exists a \emph{true} strict ordering relation among the set $S$ of $n$ input elements, but algorithms are only allowed to gather information on this relation via \emph{unreliable} comparisons between pairs of elements.
The outcome of a comparison involving two distinct elements $x$ and $y$ can be either ``$<$'' or ``$>$'' to signify that $x$  is reported to be ``smaller than'' or ``larger than'' $y$, respectively.
Most of the times the outcome of a comparison will correspond to the true relative order of the compared elements but there is a probability, which is upper bounded by a constant $p < \frac{1}{2}$, that the wrong result will be reported instead.
An algorithm can compare the same pair of elements more than once and, when this happens, the outcome of each comparison is chosen \emph{independently} of the previous results. In a similar way, comparisons involving different pairs of elements are also assumed to be independent.

The above error model 
was first considered
in the 80s and 90s when the related problems of finding the minimum, selecting the $k$-th smallest element, and of sorting a sequence have been studied \cite{FeigeRPU94,Pelc89,Pelc02}.
The best solutions are due to Feige et al. \cite{FeigeRPU94}, who provided optimal Monte Carlo algorithms having a success probability of $1-q$ and requiring $O\big(n \log \frac{1}{q}\big)$, $O\big(n\log \frac{\min\{k,n-k\}}{q}\big)$ and $O\big(n\log \frac{n}{q}\big)$ time, respectively. In the sequel we will invoke the minimum finding algorithm of \cite{FeigeRPU94} --which we name \findmin-- as a subroutine. We therefore find convenient to summarize its performances in the following:%
\begin{theorem}[\hspace{1sp}\cite{FeigeRPU94}]\label{thm-ftmin}
	Given a set $S$ of $n$ elements and a parameter $q \in (0, \frac{1}{2})$,
	Algorithm \findmin returns, in $O\left(n \log \frac{1}{q}\right)$ worst-case time, the minimum of $S$ with a probability of at least $1-q$. 
\end{theorem}

\subsection{Our Contributions}

We develop a randomized algorithm that solves $\ftmin(k)$ with probability $1-q$
in $O(\frac{n}{k} \log \frac{1}{q})$  expected  time for the whole range of values of $k \in [1, n-1]$, where $q \in (1, \frac{1}{2}]$.
Moreover, we show that the expected running time of our algorithm is asymptotically optimal for $k \le c \cdot n$ for any constant $c \le 1$ by proving that any algorithm that solves $\ftmin(k)$ with probability $1-q$ requires $\Omega(\frac{n}{k}\log \frac{1}{q})$ comparisons in expectation. 

These results are quite surprising since for $k=\Omega(\log \frac{1}{q})$
the expected running time must still be $\Omega(\frac{n}{k}\log \frac{1}{q})$ even in absence of comparison faults (indeed, any random subset of $o(\frac{n}{k}\log \frac{1}{q})$ elements does not contain any of smallest $k$ elements with a probability larger than $q$). 
In other words, comparison errors almost do not increase the computational complexity of the approximate minimum selection problem. 
In addition, we show how to additionally guarantee that the \emph{worst-case} running time of our algorithm will be $O( (\frac{n}{k} + \log \log \frac{1}{q}) \cdot \log
\frac{1}{q} )$. This implies that as soon as $k = O(\frac{n}{\log\log \frac{1}{q}})$ we can solve $\ftmin(k)$ with probability $1-q$ in the optimal \emph{worst-case} running time of $O(\frac{n}{k}\log \frac{1}{q})$.

Another way to evaluate algorithms for $\ftmin(k)$
is to instantiate $q$ consider the range of values of $k$ that they are able to handle, within a given (asymptotic) limit $T$ on their running time.
For example, if we require the algorithms to succeed \emph{with high probability} (i.e., we set $q=\frac{1}{n}$) and $T=O(n)$,
a natural $O(\frac{n}{k} \log^2 n)$-time algorithm that executes \findmin with $q=O(\frac{1}{n})$
on a randomly chosen subset of $O(\frac{n}{k}\log n)$ elements 
only works for $k = \Omega( \log^2 n )$, while our algorithm works for any $k=\Omega(\log n)$, thus exhibiting a quadratic difference in w.r.t.\ the smallest achievable $k$.
More importantly, when $T$ is $o(\log^2 n)$,
the natural algorithm cannot provide any bound on the rank of the returned element w.h.p.,
while our algorithms yield an asymptotically optimal upper bound of $O(\frac{n}{T}\log n)$ if $T$ refers to an expected running time or if  $T = \Omega(\log n \cdot \log \log n)$ and $T$ is refers to a worst-case running time.

To obtain or results we first design an $O(\frac{n}{k}\log \frac{1}{q})$-time reduction that transforms the problem of solving $\ftmin(k)$ into the problem of solving $\ftmin\left(\frac{3}{4}n\right)$ on instances having $n=\Theta(\log \frac{1}{q})$ elements.
This reduction shows that if it is possible to solve $\ftmin\big(\frac{3}{4}n\big)$ with probability $1-\Omega(q)$ in $T(n)$ time, then 
$\ftmin(k)$ can be solved in $O\big(\frac{n}{k}\log \frac{1}{q} \big) +T\left( \Theta(\log \frac{1}{q}) \right)$ time with probability $1-q$.

To obtain the expected optimal running time we can hence focus on solving $\ftmin\left(\frac{3}{4}n\right)$ in time $O(\frac{1}{q})$, with a probability of success of $1-q$.
We do so using a multi-phase process in which we iteratively consider one element at a time until we find a element that is likely to be among the smallest elements of $S$.
In this process elements are subject to a series of tests and advance from on phase to to next by passing the majority of them. Tests from different phases have exponentially decreasing error probabilities so that large elements are likely do be discarded quickly while the probability that a small elements advances from one phase to the next phase increases with the phase number.

The above approach also results in an algorithm that requires $O(\frac{n}{k} \log \frac{1}{q} + \log^2 \frac{1}{q})$ worst-case time to solve $\ftmin(k)$ with probability $1-q$, which is optimal for $k = O( \frac{n}{\log \frac{1}{q}} )$.
We improve this worst-case time complexity by designing an algorithm for $\ftmin\big(\frac{3}{4}n\big)$ that is reminiscent of knockout-style tournaments and requires only $O(n\log \frac{1}{q})$ worst-case time.
Thanks to our reduction, this improves the worst-case time required to solve $\ftmin(k)$ to $O\big(\frac{n}{k}\log n + (\log \frac{1}{q})\log\log \frac{1}{q}\big)$, which is optimal for $k = O( \frac{n}{\log\log \frac{1}{q}} )$.

\subsection{Other Related Works}
The problem of finding the \emph{exact} minimum of a collection of elements using unreliable comparisons had already received attention back in 1987 when Ravikumar et al.~\cite{RavikumarGL87} 
considered the variant in which
only up to $f$ comparisons can fail and proved that $\Theta(f n)$ comparisons are needed in the worst case. Notice that, in our case, $f= \Omega(\frac{n}{k})$ in expectation since a $(1/k)$-fraction of the elements must be compared and each comparison fails with constant probability. 
In \cite{Aigner97}, Aigner considered a \emph{prefix-bounded} probability of error $p < \frac{1}{2}$: at any point during the execution of an algorithm, at most a $p$-fraction of the past comparisons failed. Here, the situation significantly worsens as up to $\Theta(\frac{1}{1-p})^n$ comparisons might be necessary to find the minimum (and this is tight). 
Moreover, if the fraction of erroneous comparisons is \emph{globally} bounded by $\rho$, and $\rho=\Omega(\frac{1}{n})$, then Aigner also proved that no algorithm can succeed with certainty~\cite{Aigner97}.
The landscape improves when we assume that errors occur independently at random:
in addition to the already-cited $O(n \log \frac{1}{q})$-time algorithm by Feige et al.~\cite{FeigeRPU94} (see Section~\ref{sec:comparison_model}), a recent paper by Braverman et al.~\cite{BravermanMW16} also considered the \emph{round complexity} %
and the number of comparisons needed by partition and selection algorithms. The results in \cite{BravermanMW16} imply that,
for constant error probabilities, $\Theta(n \log n)$ comparisons are needed by any algorithm that selects the minimum w.h.p.

Recently, Chen et al. \cite{ChenGMS17} focused on computing the smallest $k$ elements given $r$ independent noisy comparisons between each pairs of elements. For this problem, in a more general error model, they provide a tight algorithm that requires at most $O(\sqrt{n} \,\textup{polylog}\, n)$ times as many samples as the best possible algorithm that achieves the same success probability.

If we turn our attention to the related problem of sorting with faults, then $\Omega(n \log n + fn)$ comparisons are needed to correctly sort $n$ elements when up to $f$ comparisons can return the wrong answer, and this is tight \cite{LakshmananRG91, Long92, Bagchi92}. 
In the prefix-bounded model, the result in \cite{Aigner97} on minimum selection also implies that $(\frac{1}{1-p})^{O(n \log n)}$ comparisons are sufficient for sorting, while a lower bound of $\Omega\big( ( \frac{1}{1-p} )^n \big)$ holds even for the easier problem of checking whether the input elements are already sorted \cite{BorgstromK93}.
The problem of sorting when faults are permanent (or, equivalently, when a pair of elements can only be compared once) has also been extensively studied and it exhibits connections to both the \emph{rank aggregation} problem and to the \emph{minimum feedback arc set} \cite{MakarychevMV13, Kenyon-MathieuS07, BravermanMW16, BravermanM08, KleinPSW22,LeightonM99, GeissmannMW15, GeissmannLLP17, GeissmannLLP20, GeissmannLLP19}.  Error-prone models have also been considered in the context of optimization algorithms \cite{Geissmann0LPP19} and in the design resilient data strictures \cite{FinocchiGI09, 0001LM19}.
For more related problems on the aforementioned and other fault models,
we refer the interested reader to \cite{Pelc02}  for a survey and to \cite{Cicalese13} for a monograph.

Finally, we point out that, in the fault-free case,
a simple sampling strategy allows to find one of the smallest $k$ elements in $O(\min\{n, \frac{n}{k} \log \frac{1}{q} \})$ time with probability at least $1-q$.

\subsection{Paper Organization}
In Section~\ref{sec:preliminaries} we give some preliminary remarks and we outline a simple strategy to reduce the error probability. In Section~\ref{sec:lower_bounds} we prove our lower bound, while Section~\ref{sec:reduction} is devoted to our reduction from $\ftmin(k)$ to $\ftmin(\frac{3}{4}n)$.
Finally, in Section~\ref{sec:optimal_expected} and Section~\ref{sec:worst_case}, we design our algorithms that solve $\ftmin(k)$ in $O(\frac{n}{k} \log \frac{1}{q})$ expected time and $O(\frac{n}{k} \log \frac{1}{q} + (\log \frac{1}{q}) \log\log \frac{1}{q})$ worst-case time, respectively.

\section{Preliminaries}
\label{sec:preliminaries}

We will often draw elements from the input set into one or more (multi)sets using sampling with replacement, i.e., we allow multiple copies of an element to appear in the same multiset.
We will then perform comparisons among the elements of these multisets as if they were all distinct: when two copies of the same 
element are compared, we break the tie using any arbitrary (but consistent) ordering among the copies.

According to our error model, a comparison faults happen independently at random with probability of at most $p < \frac{1}{2}$. This probability can be reduced by repeating a comparison multiple times and using a simple majority strategy.
Incidentally, the same strategy also works in the related setting in which pairwise comparisons are not longer allowed but we have access to a \emph{noisy oracle} $\mathcal{O}$ that can be queried with an element $x \in S$ and returns either true or false depending on some property of $x$. Here $\mathcal{O}$ correctly answers a query with probability at least $1-p$ and $\mathcal{O}$'s errors are independent.
The performances of this majority strategy are formalized in the following:
\begin{lemma}
	\label{lemma:prob_boosting}
	Let $x$ and $y$ be two distinct elements.
	For any error probability upper bounded by a constant $p \in [0, \frac{1}{2} )$ there exists a constant $c_p \in \mathbb{N}$ such that the strategy that compares $x$ and $y$ (resp. queries $\mathcal{O}$ with $x$) $2 c_p \cdot t+1$ times and returns the majority result is correct with probability at least $1- e^{-t}$.
\end{lemma}
\begin{proof}
	Suppose, w.l.o.g., that $x < y$. Let $X_i \in \{0,1\}$	be an indicator random variable that is $1$ iff the $i$-th comparison (resp. query) is correct.
	Since the $X_i$s are independent Bernoulli random variables of parameter at least $1-p$, $\sum_{i=1}^{2 c_p \cdot t+1} X_i$ is stochastically dominated by a binomial random variable $X$ of parameters $2 \eta = 2 c_p t +1$ and $1-p$, and hence $\E[X]=2\eta(1-p)=(1-p)(2c_p\cdot t+1)$.
	Moreover, since $p < 1/2$ we know that $2(1-p) > 1$ and hence we can use 
	the Chernoff bound~\cite[Theorem~4.2~(2)]{MitzenmacherU05}
	$
	\Pr\left( X \le (1-\delta)\E[X] \right) \le \exp\left(-\frac{\delta^2 \E[X]}{2} \right) \; \forall \delta \in (0,1)
	$
	to upper bound to the probability of failure of the majority strategy. Indeed:
	\begin{align*}    
	\Pr(X \le \eta) & = \Pr\left( X \le  \frac{1}{2(1-p)} \E[X] \right)
	\le \exp\left( - \frac{ (2 (1-p) - 1)^2}{8(1-p)^2} 2\eta (1-p) \right) \\
	&= \exp\left( - \frac{(1-2p)^2}{4(1-p)} \eta \right)
	< \exp \left(- c_p t \frac{(1-2p)^2}{4(1-p)} \right),
	\end{align*}
	which satisfies claim once we choose $c_p = \left\lceil \frac{4(1-p)}{(1-2p)^2} \right\rceil$. \qed
\end{proof}

\section{Lower Bound}
\label{sec:lower_bounds}

In this section we establish our lower bound of $\Omega(\frac{n}{k} \log \frac{1}{q})$ to the expected running time required by any algorithm that is able to solve $\ftmin(k)$ with probability $1-q$.
For asymptotically small values of $k$ we will show that any faster algorithm for $\ftmin(k)$ implies the existence of an algorithm for $\ftmin(1)$ whose running time contradicts the lower bound of \cite{FeigeRPU94}. 
For larger values of $k$ (recall that $k$ is allowed to be a function of $n$) the above strategy does not work but we can show that any algorithm that does not perform a sufficiently large number of comparisons fails to correctly identify one of the smallest $k$ elements with a sufficiently high probability when its input sequence is a random permutation of the elements.

\begin{theorem}
	The expected number of comparisons of any (possibly randomized) algorithm solving $\ftmin(k)$ with probability $1-q$, for $1 \le k \le c \cdot n$ and any constant $c<1$, is $\Omega(\frac{n}{k} \log n )$.	
\end{theorem}
\begin{proof}
	Suppose towards a contradiction that there is an algorithm $A$ that is able to solve $\ftmin(k)$ with probability $1-q$ in time $O(\frac{n}{k} \log \frac{1}{q})$, where $k = \kappa(n)$ is allowed to be a function of $n$. 

	We distinguish two cases depending on whether  $\kappa(n) = o(n)$  or $\kappa(n)=\Omega(n)$.
	If $\kappa(n) = o(n)$ then pick a value of $\lambda$ that satisfies $\lambda \ge \kappa(n \lambda)$ (this value always exists since $\kappa(n \lambda) = o(\lambda)$) and consider the \hbox{(multi-)set} $S'$ consisting of $\lambda$ copies of each element in $S$. If we run $A$ on $S'$ with $k = \lambda$ we obtain an algorithm with a running time of $o(\frac{n \lambda}{\lambda} \log \frac{1}{q}) = o(n \log \frac{1}{q})$ that returns one of the smallest $\lambda$ elements of $S'$ with probability $1-q$. Since the smallest $\lambda$ elements of $S'$ are all copies of the smallest element in $S$, the above algorithm solves $\ftmin(1)$ with probability $1-q$, violating the lower bound of $\Omega(n \log \frac{1}{q})$ on the running time of any such algorithm \cite{FeigeRPU94}.\footnote{A closer inspection of the proof of \cite{FeigeRPU94} shows that such a lower bound also applies to expected running times.}

	If $\kappa(n) = \Omega(n)$ we can assume that $q = o(1)$, as otherwise $\Omega(\frac{n}{k} \log \frac{1}{q}) = \Omega(1)$, which is a trivial lower bound.
	We consider the execution of $A$ on a random permutation $S$ of the elements in $\{1, \dots, n\}$.
	Let $\mu$ be the expected running time of $A$ on $S$ and suppose that $2\mu < \frac{\log\frac{1}{2q}}{\log \frac{1}{1-c}} = \frac{\log 2q}{\log (1-c)}$. We will say that an element $x \in S$ is \emph{accessed} by $A$ if $x$ is the returned element or if $x$ is compared with another element at some point during the execution of $A$. The number of accessed elements is a lower bound on the running time of $A$ which, by Markov inequality, is at least $2\mu$ with probability at most $\frac{1}{2}$. In other word, in at least half the executions of $A$, the number of distinct accessed elements is less than $2\mu$.
	Consider any such execution and notice that, since $S$ is a random permutation, the probability that none of the accessed elements is within the $k$ smallest elements of $S$ must be larger than $\left( \frac{(1-c)n}{n} \right)^{2\mu} \le (1-c)^{\frac{\log 2q}{\log (1-c)}} = 2q$.
	Therefore the overall failure probability of $A$ must exceed $\frac{1}{2} \cdot 2q = q$. Since this contradicts our hypothesis on the success probability of $A$, we must conclude that it is false that $2\mu < \frac{\log 2q}{\log (1-c)}$, i.e., we must have $\mu = \Omega(\log \frac{1}{q})$. \qed
	
\end{proof}

\section{Reduction}
\label{sec:reduction}

In this section we reduce the problem of solving $\ftmin(k)$ %
to the problem of solving $\ftmin(\frac{3}{4}n)$.\footnote{In order to ease the notation, here and throughout the rest of the paper we will use $\ftmin(\frac{3}{4}n)$ as a shorthand for $\ftmin(\lceil \frac{3}{4}n \rceil)$.}
Throughout the rest of the paper we will say that an element $x$ is \emph{small} if it is one of the smallest $k$ elements of $S$, otherwise we say that $x$ is \emph{large}.
The reduction selects a set $S^*$ of size $m$ containing at least $\frac{3}{4}m$ small elements, where the value of $m$ will be determined later.

\begin{itemize}[leftmargin=*]
	\item Create $m$ sets by independently sampling, with replacement, $\lceil 3 \frac{n}{k} \rceil$ elements per set from $S$.	
	\item Run \findmin with error probability $q = \frac{1}{10}$ on each of the sets. Let $S^* = \{x_1, \dots, x_m\}$ be the collection of the returned elements, where $x_i$ is the element selected from the $i$-th set.
\end{itemize}

\noindent Using Theorem~\ref{thm-ftmin}, Lemma~\ref{lemma:prob_boosting} and the Chernoff bound, we are able to prove the following:

\begin{lemma}
	\label{lemma:prob_three_fourth_m}
	The probability that less than $\frac{3}{4}m$ elements in $S^*$ are small is at most $e^{- \frac{m}{240}}$.
\end{lemma}
\begin{proof}
	Since the $i$-th set contains at least $3 \frac{n}{k}$ elements and each of them is small independently with a probability of $\frac{k}{n}$, the probability that no element in the $i$-th set is small is upper bounded by
	\[
		\left(1- \frac{k}{n}\right)^{3 \frac{n}{k}} \le \left( e^{-\frac{k}{n}} \right)^{3 \frac{n}{k}} = e^{-3} < \frac{1}{20},
	\]
	\noindent where we used the inequality $1+x \le e^x$.
	In other words, for every $i$, the event ``the $i$-th set contains a small element'' has probability at least $1-\frac{1}{20}$.
	Moreover, by our choice of $q$, the probability that \findmin returns the correct minimum of the $i$-th set is at least $1-\frac{1}{10}$.
	Clearly, if both the previous events happen, $x_i$ must be a small element and, by the union bound, the complementary probability can be at most $\frac{1}{20} + \frac{1}{10} < \frac{1}{6}$.

	Let $X_i$ be an indicator random variable that is $1$ iff $x_i$ is a small element so that 
	$X = \sum_{i=1}^m X_i$ is the number of small elements in $S^*$.
	Since the $x_i$s are independently small with a probability of at least $\frac{5}{6}$, the variable $X$ is stochastically larger than a Binomial random variable of parameters $m$ and $\frac{5}{6}$.
	As a consequence $\E[X] \ge \frac{5}{6} m$ and, by using Chernoff bound~\cite[Theorem~4.2~(2)]{MitzenmacherU05}, we obtain:	
	\begin{equation*}
	\Pr\left( X \le \frac{3}{4} m \right) \le \Pr\left(X \le \frac{9}{10} \E[X]\right) \le e^{- \frac{1}{2}(\frac{1}{10})^2 \cdot \frac{5}{6} m} =e^{- \frac{m}{240}}. \tag*{\qed}
	\end{equation*} 
\end{proof}

\noindent We are now ready to show the consequence of the above reduction:

\begin{lemma}
	\label{lemma:reduction_findsmall}
	Let $A$ be an algorithm that solves $\ftmin( \frac{3}{4}n )$ with a probability of success of at least $1-q_A$, and let $T(n,q_A)$ be its running time.
	For any $k$ and any $q \in (0, \frac{1}{2}]$, there exists an algorithm that solves $\ftmin(k)$ with probability $1-q$ in $O\left(\frac{n}{k} \log \frac{1}{q} \right) + T\left(\Theta(\log \frac{1}{q}), q/2 \right)$ time.
\end{lemma}
\begin{proof}
	We first choose $m = \lceil 240 \ln \frac{2}{q} \rceil$ and we compute the set $S^*$ according to our reduction.\footnote{Throughout the rest of the paper we will use $\log$ for binary logarithms and $\ln$ for natural logarithms.}  Then we run $A$ on $S^*$ with success probability $q_A = q/2$ and answer with the element it returns.
	The first step of the reduction can be easily implemented in $O(\frac{mn}{k}) = O(\frac{n}{k} \log \frac{1}{q})$ time, and since
	each of the $m = O(\log \frac{1}{q})$ executions of \findmin requires time $O( \frac{n}{k} \log \frac{1}{q}) = O(\frac{n}{k})$ (see Theorem~\ref{thm-ftmin} and recall that $q=1/10$), the total time spent for the second step is also $O(\frac{n}{k} \log \frac{1}{q})$. %
	By Lemma~\ref{lemma:prob_three_fourth_m}, the probability that less than $\frac{3}{4}m$ elements in $S^*$ are small is at most 
	$e^{-\frac{m}{240}} \ge
	e^{- \ln \frac{2}{q}} = \frac{q}{2}$.
	Since
	the probability that $A$ returns one of the smallest $\frac{3}{4}m$ elements in $S^*$ is at least $q_a = \frac{q}{2}$, the claim follows by using the union bound. \qed
\end{proof}

It is not hard to see that, if we choose algorithm $A$ in Lemma~\ref{lemma:reduction_findsmall} to be \findmin we have $T(n, q_A)= O(n\log \frac{1}{q_A})$ which, thanks to our reduction, allows to solve $\ftmin(k)$ in time $O(\frac{n}{k}\log \frac{1}{q} +\log^2 \frac{1}{q})$ with probability $1 - q$. This matches our lower bound of $\Omega(\frac{n}{k}\log \frac{1}{q})$ for $k = O(\frac{n}{\log 1/q})$.
Nevertheless, the major difficulty in solving $\ftmin(k)$ lies in the case $k=\omega(\frac{n}{\log 1/q})$.

\section{Solving \texorpdfstring{$\ftmin(k)$}{FT-Min(k)} in Optimal Expected Time}
\label{sec:optimal_expected}

In this section we solve $\ftmin(k)$ with probability $1-q$ in the optimal $O(\frac{n}{k} \log \frac{1}{q})$ expected time.
We achieve this by designing a simple algorithm that requires $O(\log\frac{1}{q})$ expected time to solve $\ftmin(\frac{3}{4}n)$ with probability $1-q$. We will assume that $n \ge 4$ since otherwise it suffices to return any element in $S$.

Given $\delta \in [0,1]$, we call $S^-_\delta$ the set containing the smallest $\lceil \delta n \rceil$ elements of $S'$, and we let $S^+_\delta = S \setminus S^-_\delta$. 
We will make use of an oracle $\mathcal{O}$ that can be queried with an element $x \in S$, answers in constant time, and satisfies the following conditions:
$\mathcal{O}$ reports an element $x$ to be small with probability at least $1 - \frac{2}{5}$ if $x \in S^-_{1/3}$ and with probability at most $\frac{2}{5}$ if $x \in S^+_{3/4}$. In other words, $\mathcal{O}$ identifies whether an element in $S^-_{1/3} \cup S^+_{3/4}$ is small with a failure probability of at most $\frac{2}{5}$.
Queries to $\mathcal{O}$ can be repeated and errors in the answers are independent.
Notice that $\mathcal{O}$ provides no guarantees on the success probability if $x \in S^-_{3/4} \setminus S^-_{1/3}$.

To design $\mathcal{O}$ we find convenient to assume that $p \le \frac{1}{16}$. We can do so without loss of generality thanks to Lemma~\ref{lemma:prob_boosting} (if $p > \frac{1}{16}$ then we can simulate each comparison with the majority result of $6c_p+1$ comparison).
The oracle $\mathcal{O}$ answers a query for an element $x \in S'$ by comparing $x$ with a randomly sampled element $y$ from $S \setminus \{x\}$. If $x$ compares smaller than $y$, then $x$ is reported to be small, otherwise it is reported to be large.
Suppose that $x \in S^-_{1/3}$, if $x$ is (incorrectly) reported as large at least one of the following two conditions must be true i) $y \in S^-_{1/3}$ or ii) the comparison between $x$ and $y$ returned the wrong result.
The first condition is true with probability at most $\frac{\lceil n/3 \rceil - 1}{n-1} \le \frac{1}{3}$ while the second condition is true with probability at most $p \le \frac{1}{15}$. 
Therefore the probability that $x$ is reported as large is at most $\frac{1}{3} + \frac{1}{16} < \frac{2}{5}$.
If $x \in S^+_{3/4}$ then, in order for $x$ to be (incorrectly) reported as small,
we must have that i) $y \in S^+_{3/4}$ or ii) the comparison between $x$ and $y$ returned the wrong result.
The first condition is true with probability at most $\frac{n - \lceil 3n/4 \rceil - 1}{n-1} \le \frac{1}{4}$ while the second condition is true with probability at most $p \le \frac{1}{16}$.
Overall, $x$ is reported as small with probability at most $\frac{1}{4} + \frac{1}{16} < \frac{2}{5}$.

We are now ready to describe our algorithm. Let $c_\mathcal{O}$ be the constant of Lemma~\ref{lemma:prob_boosting} for $p=\frac{2}{5}$. The algorithm works in phases: In the $i$-th phase we select one element $x_i$ uniformly at random from $S$ and we perform a \emph{test} on $x_i$. This test consists of $2 c_\mathcal{O} \lceil \ln \frac{2^i}{q} \rceil + 1$ queries to $\mathcal{O}$ and \emph{succeeds} if $x$ is reported as small at by the majority of the queries (otherwise it \emph{fails}).
If the test on $x_i$ succeeds we return $x_i$. Otherwise we move to the next phase.

We start by proving a lower bound on the success probability of our algorithm.

\begin{lemma}
	The above algorithm solves $\ftmin(\frac{3}{4}n)$ with probability at least $1-\frac{q}{4}$.
	\label{lemma:exp_optimal_success}
\end{lemma}
\begin{proof}
	If the algorithm fails then either it does not terminate or it returns an element in $S^+_{3/4}$. It is easy to see that the algorithm terminates almost surely. We now upper bound the probability that the algorithm terminates during phase $i$ by returning a large element (recall that an element if large if it is not among the $k = \lceil \frac{3}{4} n\rceil$ smallest elements of $S$).
	In order for this to happen we must have that i) $x_i$ is large and ii) was reported as small by at least $c_\mathcal{O} \lceil \ln \frac{2^i}{q} \rceil + 1$ of the $2 c_\mathcal{O} \lceil \ln \frac{2^i}{q} \rceil + 1$ queries to $\mathcal{Q}$. 
	The probability of i) is at most $\frac{1}{4}$ and,
	by Lemma~\ref{lemma:prob_boosting}, the probability
	of ii) is at most $e^{-\lceil \ln \frac{2^i}{q} \rceil} \le \frac{q}{2^i}$.
	Using the union bound over the different phases of the algorithm, we can upper bound the overall probability of failure as $\frac{1}{4}\sum_{i=1}^\infty \frac{q}{2^i} = \frac{q}{4}$. \qed
\end{proof}

\begin{lemma}
	The above algorithm has an expected running time of $O(\log \frac{1}{q})$.
	\label{lemma:exp_optimal_time}
\end{lemma}
\begin{proof}
	Consider a generic phase $i$ and suppose that the algorithm did not stop during phases $1, 2, \dots, i-1$. The probability that the algorithm stops during phase $i$ is at least:
	\[
		\Pr\left(x_i \in S^-_{\frac{1}{3}}\right) \cdot \Pr\left(\text{the test on $x_i$ succeeds} \, \Big| \, x_i \in S^-_{\frac{1}{3}}\right)
		\ge \frac{1}{3} \cdot \left(1 - \frac{q}{2^i} \right) \ge \frac{1}{4},
	\]
	where we used the fact that the test on $x_i$ succeeds with probability at least $1 - e^{- \lceil \ln \frac{2^i}{q} \rceil} \ge 1 - \frac{q}{2^i} \ge \frac{1}{4}$ as ensured by Lemma~\ref{lemma:prob_boosting}.
	
	Since phase $i$ requires time at most $\kappa \ln \frac{2^i}{q}$ for some constant $\kappa$, the expected running time is at most
	\begin{equation*}
		 \sum_{i=1}^\infty \kappa \ln \frac{2^i}{q} \cdot \left( \frac{3}{4} \right)^{i-1} \cdot \frac{1}{4} \le 
		\kappa \ln \frac{2}{q}  \cdot \sum_{i=1}^\infty \frac{i}{4} \cdot \left( \frac{3}{4} \right)^{i-1} = 4 \kappa \ln \frac{2}{q} = O\left( \log  \frac{1}{q} \right).
	\tag*{\qed}
	\end{equation*}
\end{proof}

Combining Lemma~\ref{lemma:exp_optimal_success} and Lemma~\ref{lemma:exp_optimal_time} we can conclude that our algorithm solves $\ftmin(\frac{3}{4}n)$ with probability at least $1-q$ in expected $O(\log \frac{1}{q})$ time. Lemma~\ref{lemma:reduction_findsmall} immediately implies the following:

\begin{theorem}
	$\ftmin\left( k \right)$ can be be solved with probability $1-q$ in $O( \frac{n}{k} \log \frac{1}{q} )$ expected time.
\end{theorem}

We conclude this section by pointing out that if we run the above algorithm for up to $6 \log \frac{1}{q}$ phases (and return a random element if the number of phases is exceeded) we can also upper bound its worst-case time complexity with $O( \sum_{i=1}^{6\log 1/q} \log \frac{2^i}{q}) = O(\log^2 \frac{1}{q})$ and lower bound its success probability with $1 - \frac{q}{4} - \left( \frac{3}{4} \right)^{6\log \frac{1}{q}} > 1 - \frac{q}{4} - q^2 \ge 1 - q$.
Combining this algorithm with Lemma~\ref{lemma:reduction_findsmall}, we obtain an algorithm that solves $\ftmin\left( k \right)$ with probability $1-q$ in $O( \frac{n}{k} \log \frac{1}{q} )$ expected time and $O(\frac{n}{k} \log \frac{1}{q}  + \log^2 \frac{1}{q})$ worst-case time. In the next section we will improve the worst-case running time.

\section{Solving \texorpdfstring{$\ftmin(k)$}{FT-Min(k)} in Almost-Optimal Worst-Case Time}
\label{sec:worst_case}

In this section we design an algorithm that solves $\ftmin\big(\frac{3}{4}n \big)$ with probability $1-q$ using $O(\log \frac{1}{q})$ comparisons in the worst case. 

For the sake of simplicity we assume that $n$ is a power of $2$,\footnote{If this is not the case it will suffice to pad $S$ with $2^{\lceil \log n \rceil}-n$ randomly selected copies of elements from $S$.} we will make use of a parameter $\rho \in (0, \frac{1}{2}]$, and we start by describing an algorithm that requires $O( n \cdot  \log \frac{1}{\rho}  \cdot (\log n + \log \frac{1}{\rho}))$ time to solve $\ftmin\big(\frac{3}{4}n \big)$ with a success probability of $1 - \rho^n$.

Our algorithm simulates a knockout tournament and works in $\log n$ rounds: in the beginning we construct a set $S_n$ that contains $n$ elements each of which is obtained
by running \findmin with probability of error $\frac{\rho^2}{2}$ on a set of
$\lceil 2 \log \frac{1}{\rho} \rceil$ elements that are  sampled \emph{with replacement}  from the input set $S$. Then, in the generic $i$-th round we match together $\frac{n}{2^{i}}$ pairs of elements selected \emph{without replacement} from the set $S_{\frac{n}{2^{i-1}}}$ and we add the \emph{match winners} to a new set $S_{\frac{n}{2^i}}$. 
After the $(\log n)$-th round we are left with a set $S_1$ containing a single element: this element is \emph{winner of the tournament}, i.e., it is the element returned by our algorithm.

A match between elements $x$ and $y$ in the $i$-th round consists of 
$
2 c_p \left\lceil 2^{i} \ln \frac{1}{\rho} \right\rceil + 3 
$ comparisons using the majority strategy of Lemma~\ref{lemma:prob_boosting}, i.e., the winner of the match is the element that is reported to be smaller by the majority of the comparisons (here $c_p$ is the constant of Lemma~\ref{lemma:prob_boosting}).
The following lemma provides a lower bound on the success probability of our algorithm:
\begin{lemma}
	\label{lemma:tournament_success_prob}
	Consider a tournament among $n$ elements, where $n$ is a power of $2$. The probability that the winner of the tournament is a small element is at least 
	$1 - \rho^{n+1}$.
\end{lemma}
\begin{proof}
	We prove the claim by induction on $n$ by upper bounding the complementary probability.
	If $n=1$, then there exists only one element $x \in S_{n}$,
	which has been obtained by running \findmin on a set $X$ of $\lceil \log \frac{1}{\rho} \rceil$	elements. Since each element of $X$ is large with probability at most $\frac{1}{4}$, the probability that no element in $X$ is small is at most $\frac{1}{4^{2\log \frac{1}{\rho}}} = \rho^4$. Since the invocation of \findmin  fails with probability at most $\frac{\rho^2}{2}$ we have that the $x$ is large with probability at most $\rho^4 + \frac{\rho^2}{2} \le \frac{\rho^2}{4} + \frac{\rho^2}{2} < \rho^2$.
	
	Now, let $n \ge 2$ be a power of $2$ and suppose that the claim holds for tournaments of $n/2$ elements. We prove that it must also hold for tournaments of $n$ elements.
	Let $x$ be the winner of the tournament. Since $n \ge 2$, $x$ must be the winner of the match between the two elements $x_1, x_2 \in S_2$ which, in turn, must be the winners of two (independent) sub-tournaments involving $\frac{n}{2}$ elements each.
	For $x$ to be large either (i) $x_1$ and $x_2$ are both large, which happens with probability at most $\rho^{2 \cdot \frac{n}{2} + 2} = \rho^{n+2}$ by inductive hypothesis, or (ii) exactly one of $x_1$ and $x_2$ is large and it wins the match.
	The probability that exactly one of $x_1$ and $x_2$ is large can be upper-bounded by the probability that at least one of $x_1$ and $x_2$ is large, which is at most $2 \cdot \rho^{\frac{n}{2} + 1}$. We hence focus on the probability that, in a match between a large and a small element, the large element wins.
	Since $x_1$ and $x_2$ are compared at least $2 c_p \left( \frac{n}{2}  \ln \rho \right) + 3$ times we know, by Lemma~\ref{lemma:prob_boosting},
	that this probability must be smaller than $e^{ - \frac{n}{2} \ln \frac{1}{\rho} - 1 } = \rho^{\frac{n}{2}+1}$. Putting it all together, we have:
	\begin{equation*}	
	\Pr(x \mbox{ is large}) \le \rho^{n+2} + 2 \cdot \rho^{\frac{n}{2} + 1} \cdot \rho^{\frac{n}{2}+1} = (2 \rho) \cdot \rho^{n+1} \le \rho^{n+1}. 	
\tag*{\qed}
	\end{equation*} 
\end{proof}

\noindent We now analyse the running time of our algorithm.

\begin{lemma}\label{lem-time-tournament}
	Simulating the tournament requires
	$O(n \cdot \log n  \cdot \log \frac{1}{\rho} + n \cdot \log^2 \frac{1}{\rho})$ time.
\end{lemma}
\begin{proof}
	The initial selection of the elements in $S_n$ requires time $O(n \cdot \log^2 \frac{1}{\rho})$.
	The tournament itself consists of $\log n$ rounds. The number of matches that take place in round $i$ is $\frac{n}{2^i}$ and, for each match, $O( 2^i \log \frac{1}{\rho} )$ comparisons are needed. It follows that the total number of comparisons performed in each round is $O(n \log \frac{1}{\rho})$ and, since there are $O(\log n)$ rounds,  the overall running time is $O(n \cdot \log n \cdot \log \frac{1}{\rho})$. \qed
\end{proof}

If we now select $\rho = \min \{ \frac{1}{2}, q^{\frac{1}{n}} \}$, we obtain an algorithm for $\ftmin\big(\frac{3}{4}n \big)$ with a running time of $O( n \log n + \log n \cdot \log \frac{1}{q} + \frac{\log^2 \frac{1}{q}}{n})$ and with a success probability of at least 
$1- \rho^{n+1} \ge 1- \rho \cdot (q^{\frac{1}{n}})^n = 1- \rho \cdot q \ge 1 - \frac{q}{2}$.
We can now using this algorithm in our reduction of Lemma~\ref{lemma:reduction_findsmall} to  immediately obtain an algorithm for $\ftmin(k)$ which is optimal for 
$k = O(\frac{n}{\log \log \frac{1}{q}})$.

\begin{theorem}
	$\ftmin(k)$ can be solved with probability $1-q$ in a worst-case time of $O(\frac{n}{k} \log \frac{1}{q} + (\log \frac{1}{q}) \log \log \frac{1}{q} )$. 
\end{theorem}

We can actually combine this algorithm with the one from Section~\ref{sec:optimal_expected}
to simultaneously solve $\ftmin(k)$ with probability $1-\frac{1}{q}$ in $O(\frac{n}{k} \log \frac{1}{q})$ expected time and $O(\frac{n}{k} \log \frac{1}{q} + (\log \frac{1}{q}) \log \log \frac{1}{q})$ worst-case time.

In order to do so we simply run the two algorithms for $\ftmin\big(\frac{3}{4}n \big)$ in parallel until one of them returns.
Clearly the expected running time is asymptotically unaffected, while the probability that this combined algorithm fails can be upper bounded by the sum of the respective failure probabilities, i.e., by at most $\frac{q}{4} + \frac{q}{2} < q$  (recall that the algorithm for $\ftmin\big(\frac{3}{4}n \big)$ of section \ref{sec:optimal_expected} fails with probability at most $\frac{q}{4}$, as shown by Lemma~\ref{lemma:exp_optimal_success}).

\begin{acknowledgements}
The authors wish to thank Tomáš~Gavenčiak, Barbara~Geissmann, Paolo~Penna, and \hbox{Peter~Widmayer} for insightful discussions.
We also wish to thank the anonymous reviewers and the coordinating editor for their careful reading of our manuscript and their many comments and suggestions.
\end{acknowledgements}

\bibliographystyle{spmpsci}
\bibliography{bibliography}

\begin{thebibliography}{10}
\providecommand{\url}[1]{{#1}}
\providecommand{\urlprefix}{URL }
\expandafter\ifx\csname urlstyle\endcsname\relax
  \providecommand{\doi}[1]{DOI~\discretionary{}{}{}#1}\else
  \providecommand{\doi}{DOI~\discretionary{}{}{}\begingroup
  \urlstyle{rm}\Url}\fi

\bibitem{Aigner97}
Aigner, M.: Finding the maximum and minimum.
\newblock Discrete Applied Mathematics \textbf{74}(1), 1--12 (1997)

\bibitem{Anthes13}
Anthes, G.: Inexact design: beyond fault-tolerance.
\newblock Communications of the ACM \textbf{56}(4), 18--20 (2013)

\bibitem{Bagchi92}
Bagchi, A.: On sorting in the presence of erroneous information.
\newblock Information Processing Letters \textbf{43}(4), 213--215 (1992)

\bibitem{Baumann05}
Baumann, R.C.: Radiation-induced soft errors in advanced semiconductor
  technologies.
\newblock IEEE Transactions on Device and Materials Reliability \textbf{5}(3),
  305--316 (2005)

\bibitem{BorgstromK93}
Borgstrom, R.S., Kosaraju, S.R.: Comparison-based search in the presence of
  errors.
\newblock In: Proceedings of the Twenty-fifth Symposium on Theory of Computing
  ({STOC}93), pp. 130--136 (1993)

\bibitem{BravermanMW16}
Braverman, M., Mao, J., Weinberg, S.M.: Parallel algorithms for select and
  partition with noisy comparisons.
\newblock In: Proceedings of the Forty-eighth48th Symposium on Theory of
  Computing ({STOC}16), pp. 851--862 (2016)

\bibitem{BravermanM08}
Braverman, M., Mossel, E.: Noisy sorting without resampling.
\newblock In: Proceedings of the Nineteenth Symposium on Discrete Algorithms
  ({SODA}08), pp. 268--276 (2008)

\bibitem{Catania04}
Catania, J.A.: Soft errors in electronic memory – a white paper (2004)

\bibitem{CheemavalaguKP04}
Cheemavalagu, S., Korkmaz, P., Palem, K.: Ultra low-energy computing via
  probabilistic algorithms and devices: {CMOS} device primitives and the
  energy-probability relationship.
\newblock In: Proceedings of the 2004 International Conference on Solid State
  Devices and Materials, pp. 402--403 (2004)

\bibitem{CheemavalaguKPAC05}
Cheemavalagu, S., Korkmaz, P., Palem, K., Akgul, B.E.S., Chakrapani, L.N.: A
  probabilistic {CMOS} switch and its realization by exploiting noise.
\newblock In: Proceedings of the 2005 IFIP/IEEE International Conference on
  Very Large Scale Integration - System on a Chip ({VLSI-SoC}05, pp. 535--541
  (2005)

\bibitem{ChenGMS17}
Chen, X., Gopi, S., Mao, J., Schneider, J.: Competitive analysis of the top-$k$
  ranking problem.
\newblock In: Proceedings of the Twenty-Eighth Symposium on Discrete Algorithms
  ({SODA}17), pp. 1245--1264 (2017)

\bibitem{Cicalese13}
Cicalese, F.: Fault-Tolerant Search Algorithms - Reliable Computation with
  Unreliable Information.
\newblock Monographs in Theoretical Computer Science. Springer (2013)

\bibitem{FeigeRPU94}
Feige, U., Raghavan, P., Peleg, D., Upfal, E.: Computing with noisy
  information.
\newblock {SIAM} Journal on Computing \textbf{23}(5), 1001--1018 (1994)

\bibitem{FinocchiGI09}
Finocchi, I., Grandoni, F., Italiano, G.F.: Optimal resilient sorting and
  searching in the presence of memory faults.
\newblock Theor. Comput. Sci. \textbf{410}(44), 4457--4470 (2009).
\newblock \doi{10.1016/j.tcs.2009.07.026}

\bibitem{GeissmannLLP17}
Geissmann, B., Leucci, S., Liu, C., Penna, P.: Sorting with recurrent
  comparison errors.
\newblock In: Proceedings of the Twenty-Eighth International Symposium on
  Algorithms and Computation ({ISAAC}17), pp. 38:1--38:12 (2017)

\bibitem{GeissmannLLP19}
Geissmann, B., Leucci, S., Liu, C., Penna, P.: Optimal sorting with persistent
  comparison errors.
\newblock In: Proceedings of the Twenty-seventh European Symposium on
  Algorithms ({ESA}19), pp. 49:1--49:14 (2019)

\bibitem{GeissmannLLP20}
Geissmann, B., Leucci, S., Liu, C., Penna, P.: Optimal dislocation with
  persistent errors in subquadratic time.
\newblock Theory Comput. Syst. \textbf{64}(3), 508--521 (2020).
\newblock \doi{10.1007/s00224-019-09957-5}

\bibitem{Geissmann0LPP19}
Geissmann, B., Leucci, S., Liu, C., Penna, P., Proietti, G.: Dual-mode greedy
  algorithms can save energy.
\newblock In: Proceedings of the 30th International Symposium on Algorithms and
  Computation ({ISAAC}19), \emph{LIPIcs}, vol. 149, pp. 64:1--64:18. Schloss
  Dagstuhl - Leibniz-Zentrum f{\"{u}}r Informatik (2019).
\newblock \doi{10.4230/LIPIcs.ISAAC.2019.64}

\bibitem{GeissmannMW15}
Geissmann, B., Mihal{\'{a}}k, M., Widmayer, P.: Recurring comparison faults:
  Sorting and finding the minimum.
\newblock In: Proceedings of the Twentieth International Symposium on
  Fundamentals of Computation Theory ({FCT}15), pp. 227--239 (2015)

\bibitem{Kenyon-MathieuS07}
Kenyon-Mathieu, C., Schudy, W.: How to rank with few errors.
\newblock In: Proceedings of the Thirty-nineth Symposium on Theory of Computing
  ({STOC}07), pp. 95--103 (2007)

\bibitem{KleinPSW22}
Klein, R., Penninger, R., Sohler, C., Woodruff, D.P.: Tolerant algorithms.
\newblock In: Proceedings of the Nineteenth European Symposium on Algorithms
  ({ESA}11), pp. 736----747 (2011)

\bibitem{LakshmananRG91}
Lakshmanan, K.B., Ravikumar, B., Ganesan, K.: Coping with erroneous information
  while sorting.
\newblock {IEEE} Transactions on Computers \textbf{40}(9), 1081--1084 (1991)

\bibitem{LeightonM99}
Leighton, T., Ma, Y.: Tight bounds on the size of fault-tolerant merging and
  sorting networks with destructive faults.
\newblock {SIAM} Journal on Computing \textbf{29}(1), 258--273 (1999)

\bibitem{0001LM19}
Leucci, S., Liu, C., Meierhans, S.: Resilient dictionaries for randomly
  unreliable memory.
\newblock In: Proceedings of the 27th Annual European Symposium on Algorithms,
  ({ESA}19), \emph{LIPIcs}, vol. 144, pp. 70:1--70:16. Schloss Dagstuhl -
  Leibniz-Zentrum f{\"{u}}r Informatik (2019).
\newblock \doi{10.4230/LIPIcs.ESA.2019.70}

\bibitem{Long92}
Long, P.M.: Sorting and searching with a faulty comparison oracle.
\newblock Tech. rep., University of California at Santa Cruz (1992)

\bibitem{MakarychevMV13}
Makarychev, K., Makarychev, Y., Vijayaraghavan, A.: Sorting noisy data with
  partial information.
\newblock In: Proceedings of the Fourth Conference on Innovations in
  Theoretical Computer Science ({ITCS}13), pp. 515--528 (2013)

\bibitem{MitzenmacherU05}
Mitzenmacher, M., Upfal, E.: Probability and computing - randomized algorithms
  and probabilistic analysis.
\newblock Cambridge University Press (2005)

\bibitem{PalemL13}
Palem, K., Lingamneni, A.: Ten years of building broken chips: The physics and
  engineering of inexact computing.
\newblock {ACM} Transactions on Embedded Computing Systems \textbf{12}(2s),
  87:1--87:23 (2013)

\bibitem{Pelc89}
Pelc, A.: Searching with known error probability.
\newblock Theoretical Computer Science \textbf{63}(2), 185--202 (1989)

\bibitem{Pelc02}
Pelc, A.: Searching games with errors - fifty years of coping with liars.
\newblock Theoretical Computer Science \textbf{270}(1-2), 71--109 (2002)

\bibitem{RavikumarGL87}
Ravikumar, B., Ganesan, K., Lakshmanan, K.B.: On selecting the largest element
  in spite of erroneous information.
\newblock In: Proceedings of the fourth Symposium on Theoretical Aspects of
  Computer Science ({STACs}87), pp. 88--99 (1987)

\end{thebibliography}

\end{document}